\documentclass[runningheads]{llncs}

\usepackage[utf8]{inputenc}

\usepackage{amsmath,amssymb}

\usepackage[vlined,ruled]{algorithm2e}

\DontPrintSemicolon
\newcommand{\vertex}[2]{\ensuremath{v^{(#1)}_{#2}}}
\newcommand{\pos}[1]{\ensuremath{\textrm{pos}[#1]}}
\newcommand{\pred}[1]{\ensuremath{\textrm{pred}[#1]}}
\renewcommand{\align}[1]{\ensuremath{\textrm{align}[#1]}}
\renewcommand{\root}[1]{\ensuremath{\textrm{root}[#1]}}
\newcommand{\sink}[1]{\ensuremath{\textrm{sink}[#1]}}
\newcommand{\shift}[1]{\ensuremath{\textrm{shift}[#1]}}
\newcommand{\neighborings}[1]{\ensuremath{\textrm{neighborings}[#1]}}
\newcommand{\layer}[1]{\ensuremath{\textrm{layer}[#1]}}

\usepackage{xcolor}
\newcommand{\correction}{\color{red!50!black}}

\usepackage{graphicx}\graphicspath{{./figs/}}

\begin{document}
\title{Erratum:\\\mbox{Fast and Simple Horizontal Coordinate Assignment}}
\author{Ulrik Brandes\inst{1,}\orcidID{0000-0002-1520-0430}
\and Julian Walter\inst{2}
\and Johannes Zink\inst{2,}\orcidID{0000-0002-7398-718X}}
\authorrunning{U.~Brandes, J.~Walter and J.~Zink}
\institute{ETH Z\"urich, Zürich, Switzerland. \email{ubrandes@ethz.ch}
\and Universität Würzburg, Würzburg, Germany. \email{zink@informatik.uni-wuerzburg.de}}
\maketitle
\begin{abstract}
We point out two flaws in the algorithm of Brandes and Köpf (Proc.\ GD~2001),
which is often used for the horizontal coordinate assignment in Sugiyama's framework for layered layouts. 
One of them has been noted and fixed multiple times,
the other has not been documented before and requires a non-trivial adaptation.
On the bright side, neither running time nor extensions of the algorithm
are affected adversely.

\keywords{Sugiyama layout \and Layered graph drawing \and Horizontal compaction.}
\end{abstract}

\section{Introduction}\label{sec:intro}

Horizontal coordinate assignment is the final phase
of Sugiyama's framework for drawing directed graphs~\cite{sugiyama_etal:81,hn-hgda-13}.
By that point, the original graph has been transformed into a directed acyclic graph
with its vertices assigned to a sequence of layers,
and ordered within these layers. 
Moreover, the layer assignment is proper,
i.e., edges spanning layers are subdivided by dummy vertices
so that two vertices are adjacent only if they are assigned to neighboring layers.

Brandes and Köpf~\cite{bk-fshca-02} proposed a simple algorithm
to determine $x$-coordinates for the vertices such that
their order is preserved and
the distance between each two vertices is at least some~$\delta > 0$.
The algorithm runs in time linear in the size of the input
and is implemented and taught widely.

The original publication, however, contains two mistakes.
The first is well known and easy to fix.
The other, however, is not commonly known, less straightforward,
and, to the best of our knowledge, no solution has been proposed.

The purpose of this short paper is
to identify clearly where these two mistakes occur,
and to provide suitable corrections. 
Our proposed corrections do not alter the behavior of the algorithm
and do not clash with extensions that have been proposed to accommodate,
for instance, vertices of varying size, ports, and other constraints~\cite{b-rasf-07,rsch-sphnca-15}.

\section{Background}\label{sec:recap}

In this section we sketch the algorithm of Brandes and Köpf.
Two mistakes, both occurring in the same subroutine,
are identified and corrected in the next section.

In the horizontal coordinate assignment problem,
we are given a layered directed acyclic graph $G=(V\uplus B,E;L)$,
i.e., a directed acyclic graph $G=(V\cup B,E)$ 
with two types of vertices
(original vertices and potential bend points of edges)
and a sequence of layers $L=(L_1,\ldots,L_h)$ 
that forms an ordered partition of $V\cup B$.
The layering is required to be proper,
i.e., adjacent vertices are in neighboring layers.
Moreover, the vertices are ordered
within each layer $L_i=\{\vertex{i}{1},\ldots,\vertex{i}{|L_i|}\}$,
where $\vertex{i}{k}$ denotes the $k$th vertex in the $i$th layer.
Thus, the topology of the drawing has already been fixed,
and the horizontal coordinate assignment algorithm is now to determine
the shape and the final metrics.

\begin{figure}
\centering
\includegraphics[width=0.75\linewidth]{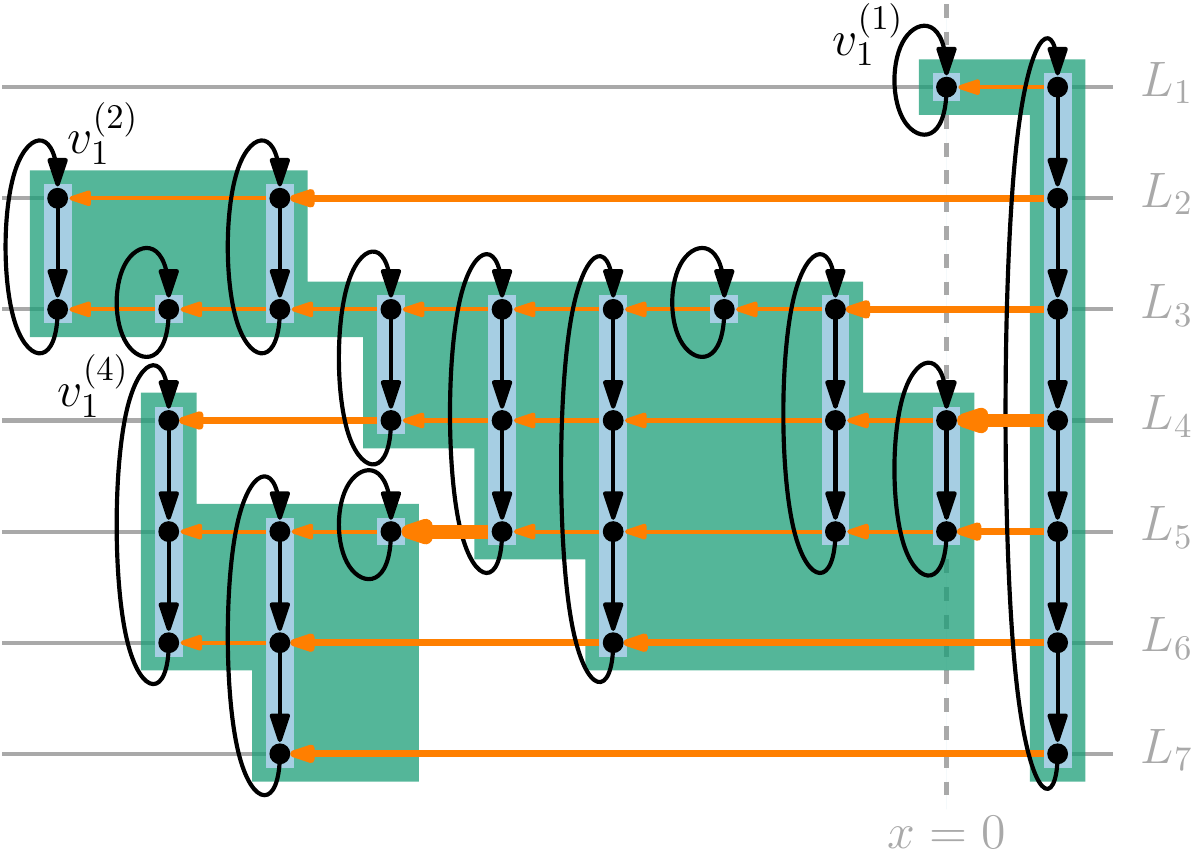}\\
\caption{Example illustrating the various definitions.
The vertices in $V\cup B$ together with the alignment (black) and predecessor (orange) relations
form a directed graph, the strongly connected components of which are the blocks (light blue rectangles). Blocks therefore form a DAG and are assigned to a class (green polygon) defined by the highest sink (labeled vertex) they can reach.}
\label{fig:example}
\end{figure}

The algorithm of Brandes and Köpf consists of the following three steps,
the first two of which are repeated four times,
once for each combination of vertical and horizontal orientations. 
We describe the top-to-bottom, left-to-right case
in which vertices are left-aligned conditional on upper neighbors;
the other three cases are symmetric.
Figure~\ref{fig:example} illustrates the subsequent definitions.
\begin{enumerate}
\item \emph{Vertical alignment:}
  Every vertex $v\in V\cup B$ is assigned a vertex, $\align{v}$,
  that can be~$v$ itself or otherwise it is either a neighbor on the next lower level
  or the uppermost vertex when tracing back these alignment pointers.
  This partitions the vertex set into singly-linked (possibly singleton) cycles
  called \emph{blocks}, where the topmost vertex in a block is called the \emph{root}.
  All vertices in a block ultimately receive the same $x$-coordinate
  in order to straighten long edges (by aligning subdivision vertices in~$B$)
  and to reduce edge lengths (by aligning vertices in $V$ with a median neighbor
  in the layer above). 
\item \emph{Horizontal compaction:} Next, vertices are aligned to the left. 
  By adding an edge between every vertex and its predecessor (left neighbor), if present,
  we obtain a directed acyclic graph of blocks.
  The sinks of this DAG are blocks whose roots are leftmost in their layer
  and are hence ordered vertically. 
  Let the sink of a vertex be the highest sink reachable from any vertex in its block by the predecessor and alignment edges.
  All vertices and blocks with the same highest sink form a \emph{class}.
  In Figure~\ref{fig:example}, there are three sinks defining a class, namely $v_1^{(1)}$, $v_1^{(2)}$, and $v_1^{(4)}$.
  The vertices in a class have a unique minimum $x$-coordinate relative to their common sink,
  and it can be obtained from longest-path layering.
  We say two classes are \emph{neighboring} each other if they contain vertices on the same layer.
  If a class is not neighboring another class with a higher sink,
  its own sink is placed at $x$-coordinate zero.
  In Figure~\ref{fig:example}, this is the case for the class with sink $v_1^{(1)}$.
  Otherwise it is placed as close to higher classes as possible.
  In Figure~\ref{fig:example}, this is the case for the two other classes.
\item \emph{Balancing:} The four $x$-coordinates per vertex
  obtained from repeating the previous steps with different orientations
  are combined into one by aligning the four resulting drawings and taking the average of the two median coordinates per vertex. 
\end{enumerate}
Both problems occur in the horizontal compaction step.
We therefore reproduce the pseudo-code of Alg.~3 from the original publication. 
It does not require any special data structures other than the vertex arrays
listed in Tab.~\ref{tab:vars}.

\begin{table}
\caption{All auxiliary data is stored in vertex arrays.}
\label{tab:vars}
\centering
\begin{tabular}{|r@{\quad}|@{\quad}l|}\hline
$\vertex{i}{k}$ & $k$th vertex in $i$th layer\\[0.5ex]\hline
\bf array & \bf content\\\hline
$\pos{v}$ & rank in layer, i.e., $\pos{\vertex{i}{k}}=k$\\
$\pred{v}$ & preceding vertex in same layer, i.e., $\pred{\vertex{i}{k}}=\vertex{i}{k-1}$\\
$\align{v}$ & vertex below in same block (or root for lowest)\\
$\root{v}$ & highest vertex in same block\\
$\sink{v}$ & if $v$ is a root:\enskip highest leftmost root, defining the class\\
$\shift{v}$ & if $v$ is a sink:\enskip offset by which class is to be shifted\\
$x[v]$ & horizontal coordinate (relative to sink, later absolute)\\\hline
\end{tabular}
\end{table}

\begin{algorithm}[t]
\renewcommand{\thealgocf}{3}
\caption{Horizontal compaction}
\BlankLine
\KwSty{function} place\_block($v$)\\
\Begin
{ \If{$x[v]$ undefined}
  { $x[v]\gets0$;\quad $w\gets v$\;
    \Repeat{$w=v$}
    { \If{$\pos{w}>1$}
      { $u\gets\root{\pred{w}}$\;
        place\_block($u$)\;
        \lIf{$\sink{v}=v$}{$\sink{v}\gets\sink{u}$}
        \uIf{$\sink{v}\neq\sink{u}$}{{\correction\nlset{(A)}\label{ln:offset}}$\shift{\sink{u}}\gets\min\{ \shift{\sink{u}},\, x[v] - x[u] - \delta \}$\;}
        \Else{$x[v]\gets\max\{ x[v],\,x[u]+\delta\}$\;}
      }
      $w\gets\align{w}$\;
    }
  }
}
\BlankLine
initialize $\sink{v}\gets v$, $v\in V\cup B$\;
initialize $\shift{v}\gets\infty$, $v\in V\cup B$\;
initialize $x[v]$ to be undefined, $v\in V\cup B$\;
\BlankLine
{\scriptsize // root coordinates relative to sink}\\
\ForEach{$v\in V\cup B$}{\lIf{$\root{v}=v$}{place\_block($v$)}}
\BlankLine
{\scriptsize // absolute coordinates}\\
\ForEach{$v\in V\cup B$}
{ $x[v]\gets x[\root{v}]$\;
  {\correction\nlset{(S)}\label{ln:shift}}\If{$\shift{\sink{\root{v}}}<\infty$}
  {$x[v]\gets x[v]+\shift{\sink{\root{v}}}\;$}
}
\end{algorithm}

\section{Problems and Corrections}\label{sec:corrections}

The problematic lines of pseudo-code concern the computation 
of absolute coordinates and the accumulation of shift values for classes.
They are marked~\ref{ln:shift} and~\ref{ln:offset} in Alg.~3.

\paragraph{Double shifting.}
The easier problem manifests itself in line~\ref{ln:shift} of Alg.~3.
To determine absolute coordinates, we want to combine, for every vertex,
the relative coordinate of its block with respect to the sink of its class (stored at the root),
with the offset of the entire class (stored at the sink).
The way Alg.~3 is formulated, offsets are actually added twice
for vertices whose root has already been shifted to its final $x$-coordinate.

Multiple implementors have noted and fixed the problem in different ways.
The correction we provide in Alg.~3a
involves aligning blocks in the recursive subroutine \textsf{place\_block}
as soon as the relative coordinate of the root with respect to the sink has been determined.
In this way, block alignment and shifting are separated and 
only offsets have to be added to every vertex in the last line of Alg.~3b.
Given the name of the subroutine, this should have been the procedure all along. 
And to simplify later steps, we also assign the sink to vertices, not just roots.

\begin{algorithm}[t]
\renewcommand{\thealgocf}{3a}
\caption{\label{alg:placeblock}Placing a block (not just the root)}
\BlankLine
\KwSty{function} place\_block($v$)
\SetKwBlock{Begin}{}{end}
\Begin
{ \If{$x[v]$ undefined}
  { $x[v]\gets0$\,;\quad $w\gets v$\;
    \Repeat{$w=v$}
    { \If{$\pos{w}>1$}
      { $u\gets\root{\pred{w}}$\;
        place\_block($u$)\;
        \lIf{$\sink{v}=v$}{$\sink{v}\gets\sink{u}$}
        \lIf{$\sink{v}=\sink{u}$}
            {$x[v]\gets\max\{ x[v],\,x[u]+\delta\}$}
        {\correction{\scriptsize // class offsets (shift of sinks) are calculated later}}
      }
      $w\gets\align{w}$\;
    }
    {\correction{\scriptsize // align the whole block}\\
    \While{$\align{w}\neq v$}{$w\gets\align{w}$\; $x[w]\gets x[v]$\; $\sink{w}\gets\sink{v}$}}
  }
}
\end{algorithm}

\paragraph{Shift accumulation.}
Problem~\ref{ln:offset} is a serious flaw that has not been documented before.
It is not clear how often it has actually resulted in visually notable problems,
but it is easy to craft instances in which it becomes severe.

During compaction, shift values are determined
for preceding classes relative to the current class.
The implicit assumption in line~\ref{ln:offset} is, however,
that the current class is not shifted itself.
As a consequence, shift values are not accumulated along critical paths in the DAG of classes.
It is not sufficient to consider the shift value of $\sink{v}$ in line~\ref{ln:offset}
because it may not be in its final state, yet.
The correction we propose in Alg.~3b therefore
separates the determination of shift values
from the placement of blocks relative to the sink of their class.

A high-level solution is to construct the DAG of classes
and perform another longest-path layering,
this time with respect to the sources in the DAG,
by placing each class as close as possible to its right neighbors.
The following lemma implies that classes are stacked diagonally, top-to-bottom, right-to-left.
It is therefore possible to accumulate shift values in this order
without recursion or explicit representation of all preceding classes.

\begin{lemma}
Let $\vertex{i}{k_1},\vertex{i}{k_2}$
be two vertices in the same layer $L_i$, $1\leq i\leq h$
with $\sink{\vertex{i}{k_1}}=\vertex{h_1}{1}$ and $\sink{\vertex{i}{k_2}}=\vertex{h_2}{1}$
after all blocks have been placed. 
If $k_1<k_2$, then $h_2\leq h_1$.
\end{lemma}
\begin{proof}
In the directed graph consisting of the alignment and predecessor relations,
every vertex can reach all sinks (and maybe more) that any vertex further to the left can reach.
The highest reachable sink is therefore at least as high as that of any vertex to the left.
\end{proof}

The monotonic ordering of classes can be utilized as follows.
Instead of ensuring that all critical paths determining the total shifts of the classes
have already been traversed by recursion
(as in the placement of blocks)
we can now start from the class with the highest sink
because we know it is rightmost in each layer.
For each sink, from top to bottom, we trace the lower contour of each class,
and make sure that the preceding classes are shifted such that
a minimum separation guarantees they do not overlap.
This is similar to the shifting of subtrees
in the tree-drawing algorithm of Reingold and Tilford~\cite{reingold/tilford:81}.

\begin{algorithm}[t]
\renewcommand{\thealgocf}{3b}
\caption{Horizontal compaction}
initialize $\sink{v}\gets v$, $v\in V\cup B$\;
initialize $\shift{v}\gets\infty$, $v\in V\cup B$\;
initialize $x[v]$ to be undefined, $v\in V\cup B$\;
\BlankLine
{\scriptsize // coordinates relative to sink}\\
\ForEach{$v\in V\cup B$}{\lIf{$\root{v}=v$}{place\_block($v$)}}
\BlankLine
{\correction
{\scriptsize // class offsets}\\
\For{$i=1,\ldots,h$}
{ \If{$\sink{\vertex{i}{1}}=\vertex{i}{1}$}
  { \lIf{$\shift{\sink{\vertex{i}{1}}}=\infty$}{$\shift{\sink{\vertex{i}{1}}}\gets0$}
    $j\gets i$\,;\quad$k\gets1$\;
    \Repeat{$k>|L_j|$ \KwSty{or} $\sink{v}\neq\sink{\vertex{j}{k}}$}
    { $v\gets\vertex{j}{k}$\;
      \While{$\align{v}\neq\root{v}$}
      { $v\gets\align{v}$\,;\quad$j\gets j+1$\;
        \If{$\pos{v}>1$}
        { $u\gets\pred{v}$\;
          \mbox{\color{black}$\shift{\sink{u}}\gets\min\{\shift{\sink{u}},\, {\correction\shift{\sink{v}}+}x[v]-(x[u]+\delta) \}$}
        }
      }
      $k\gets\pos{v}+1$
    }
  }
}}
\BlankLine
{\scriptsize // absolute coordinates}\\
\lForEach{$v\in V\cup B$}{\correction$x[v]\gets x[v]+\shift{\sink{v}}$}
\end{algorithm}

The implementation shown in Alg.~3b traverses the layers from top to bottom.
If the leftmost vertex is a sink that still has an infinite shift value,
the class is not to the left of any previously processed class with a higher sink.
It can therefore be aligned without any offset, i.e., a shift value of zero.
The lower contour of a class is traced by either moving down in a block (using $\align{v}$) 
or, if we have reached the lower end of a block pointing back at the root, 
moving one position to the right. 
The tracing ends when there is no right neighbor
or if the right neighbor is part of a class with a higher sink.

While traversing the lower contour, 
a predecessor~$u$ of the current vertex~$v$, if any,
belongs to a class with a lower sink.
We therefore update the shift value of $\sink{u}$
to ensure minimum separation between the classes of~$u$ and~$v$.
Since we are starting with an upper bound and proceed from right to left,
shift values can only decrease,
pushing classes further to the left after their initial placement.

\section{Discussion}\label{sec:concl}

We identified and corrected two flaws in the compaction step
of the well-known horizontal coordinate assignment algorithm
of Brandes and Köpf~\cite{bk-fshca-02}.
The corrections are internal to the algorithm, 
they do not alter running time in any substantial way,
and are unlikely to interfere with any of the algorithm's extensions.  

The first mistake was the result of a discrepancy between the implementation
and an attempt to make the pseudo-code more readable.
Others have identified and corrected the issue before,
notably Florian Fischer (Diplom~2004, Passau)
and John Julian Carstens (M.Sc.~2012, Kiel) in their theses.
A frequently proposed correction
only checks whether the vertex to be shifted is a root,
and therefore relies implicitly on the top-to-bottom processing order
to ensure that roots are shifted
before the other members of a block inherit their coordinate.
We believe that the correction proposed here is more appropriate.

Walter and Zink recently noted and clearly identified problem~\ref{ln:offset}.
We developed multiple corrections,
but present only one that aligns most closely with the original algorithm.
We shortly describe one alternative version of Alg.~3b in the appendix,
also because it operates more explicitly on the DAG of classes.
These corrections were developed without realizing that 
the problem had already been identified 
by Damon Lundin while working for IBM Blueworks Live (personal communication, 2013).
He even proposed a fix that bears similarity with our independently developed solutions
but does not work in general.
Unfortunately, the conversation was only recovered while preparing this contribution.

\bibliographystyle{splncs04}
\bibliography{references}
\appendix

\section*{Appendix}
\label{appendix}
Since sinks (and thus class membership) have already been determined,
we do not actually have to traverse contours to identify vertices 
whose predecessors (left neighbors) are in another class.
Instead, we can iterate over all vertices and compare their sink with that of the predecessor.
By storing the pairs of neighboring vertices in different classes, 
we are essentially constructing the DAG of classes (including its parallel edges).

An alternative approach for Alg.~3b is therefore the following.
First, class adjacencies are recorded in a list~$\neighborings{i}$ for each sink~$v_1^{(i)}$.
Shift values are then propagated down from the higher to lower classes.
Note that this may not be done during the first enumeration 
because all neighborings of a higher class need to be processed 
before propagating values from its neighboring classes.

\begin{algorithm}
        \renewcommand{\thealgocf}{3b (alternative)}
        \caption{Horizontal compaction}
        initialize $\sink{v}\gets v$, $v\in V\cup B$\;
        initialize $\shift{v}\gets\infty$, $v\in V\cup B$\;
        initialize $\layer{v^{(i)}} \gets i$, $v^{(i)}\in V\cup B$\;
        initialize $x[v]$ to be undefined, $v\in V\cup B$\;
        \BlankLine
        {\scriptsize // coordinates relative to sink}\\
        \ForEach{$v\in V\cup B$}{\lIf{$\root{v}=v$}{place\_block($v$)}}
        \BlankLine
        {\correction
        {\scriptsize // class offsets}\\
        initialize $\neighborings{i} \gets \emptyset$, $i \in \{1, \dots, 
        h\}$\;
        \BlankLine
        {\scriptsize // find all neighborings}\\
        \For{$i \gets 1, \dots, h$} {
                \For{$j \gets |L_i|, \dots, 2$} {
                        \If{$\sink{\vertex{i}{j-1}} \ne \sink{\vertex{i}{j}}$} {
                                $\neighborings{\layer{\sink{i\vertex{i}{j}}}} \gets
                                \neighborings{\layer{\sink{\vertex{i}{j}}}} \cup \{ (\vertex{i}{j-1}, \vertex{i}{j}) \}$\;
                        }
                }
        }
        \BlankLine
        {\scriptsize // apply shift for all neighborings}\\
        \For{$i \gets 1, \dots, h$} {
                \lIf{$\shift{\sink{\vertex{i}{1}}}=\infty$}{$\shift{\sink{\vertex{i}{1}}}\gets0$}
                \For{$(u, v) \in \neighborings{i}$}{
                        \mbox{\color{black}$\shift{\sink{u}}\gets\min\{\shift{\sink{u}},\, {\correction\shift{\sink{v}}+}x[v]-(x[u]+\delta) \}$}
                }
        }
        }
        \BlankLine
        {\scriptsize // absolute coordinates}\\
        \lForEach{$v\in V\cup B$}{\correction$x[v]\gets x[v]+\shift{\sink{v}}$}
\end{algorithm}

\end{document}